\documentclass[11pt,a4paper]{article}
\usepackage[T1]{fontenc}
\usepackage{amsmath,amssymb,amsthm,mathtools}
\usepackage{newtxtext,newtxmath}
\usepackage[left=26mm,right=26mm,top=27mm,bottom=27mm]{geometry}
\usepackage{microtype,array,booktabs,enumitem,titlesec}
\usepackage{tikz}
\usetikzlibrary{arrows.meta,positioning}
\usepackage[hidelinks]{hyperref}
\hypersetup{
 pdftitle={A Factor-Two Lower Bound for Complementing Two-Way Deterministic Finite Automata},
 pdfauthor={Sudais Arif},
 pdfsubject={State complexity of two-way deterministic complementation},
 pdfkeywords={two-way finite automata, complementation, state complexity, alternating groups, finite semigroups}
}
\titleformat{\section}{\normalfont\normalsize\bfseries}{\thesection.}{.5em}{}
\titleformat{\subsection}{\normalfont\normalsize\itshape}{\thesubsection.}{.5em}{}
\titlespacing*{\section}{0pt}{2.2ex plus .5ex minus .2ex}{1.1ex plus .2ex}
\titlespacing*{\subsection}{0pt}{1.8ex plus .4ex minus .2ex}{.7ex plus .2ex}
\newtheorem{theorem}{Theorem}[section]
\newtheorem{lemma}[theorem]{Lemma}

\newtheorem{corollary}[theorem]{Corollary}
\theoremstyle{definition}

\theoremstyle{remark}

\newcommand{\Alt}{\mathfrak A}
\newcommand{\Sym}{\mathfrak S}
\newcommand{\LL}{\mathsf L}
\newcommand{\RR}{\mathsf R}
\newcommand{\statecomplexity}{\operatorname{sc}}
\newcommand{\im}{\operatorname{Im}}

\newcommand{\id}{\operatorname{id}}
\newcommand{\acc}{\mathsf{accept}}
\newcommand{\rej}{\mathsf{reject}}
\renewenvironment{abstract}
 {\par\noindent\rule{\linewidth}{.4pt}\par\smallskip
  \noindent\textbf{Abstract}\par\smallskip\noindent\ignorespaces}
 {\par\smallskip\noindent\rule{\linewidth}{.4pt}\par\medskip}
\title{A Factor-Two Lower Bound for Complementing\\Two-Way Deterministic Finite Automata}
\author{Sudais Arif\\[.3em]\small\itshape Carnegie Mellon University in Qatar}
\date{}
\begin{document}
\maketitle
\thispagestyle{empty}

\begin{abstract}
For every $n\ge22$, we construct a language over a fixed three-letter
alphabet whose two-way deterministic state complexity is $n$ and whose
complement has state complexity $2n$. Acceptance is at the left
endmarker, and infinite computations reject. The lower bound holds
for arbitrary deterministic recognizers. The proof uses the group of
transformations induced by the witness to constrain the behavior of a
smaller recognizer, then derives a contradiction from a permutation
query problem. A matching halting complement follows from a
$2n$-state construction for sources whose destination states determine
the direction of each transition. Splitting states by incoming
direction gives a $2n+2b$ upper bound for general sources, where $b$
counts states entered in both directions. In the acceptance-anywhere
model of Geffert, Mereghetti, and Pighizzini, the same family has exact
complexities $N$ and $2N-1$ for every $N\ge23$.
\end{abstract}

\section{Introduction}\label{sec:intro}

Complementation is a basic operation on formal languages. If
$L\subseteq\Sigma^*$, then the complement of $L$, relative to the
fixed alphabet $\Sigma$, is $\overline L=\Sigma^*\setminus L$.
Given a machine for $L$, we ask how large a machine of the same type
for $\overline L$ must be, when such a machine exists.

For a deterministic machine that halts on every input, complementation
is obtained by exchanging accepting and nonaccepting outcomes. The
situation changes for recognizers that may have infinite computations.
A word can then fail to be accepted either because the computation
halts without accepting or because it never halts. A complementing
construction must account for both possibilities. Sipser's method
provides a way of converting nonhalting behavior into halting behavior
in space-bounded machines~\cite{Sipser1980}.

For complete one-way deterministic finite automata, complementation
is immediate: the machine reads the input once, always halts, and so
swapping final and nonfinal states recognizes the complement. Two-way
deterministic finite automata (2DFA) still recognize only regular
languages~\cite{RabinScott1959,Shepherdson1959}, but their head can
move in both directions. A 2DFA may therefore reject by looping rather
than by halting, so complementation cannot in general be obtained
simply by swapping accepting and nonaccepting states.

When Sipser's method is specialized to two-way finite automata, it
gives a quadratic $O(n^2)$ upper bound for complementing an $n$-state
2DFA. Geffert, Mereghetti, and Pighizzini (GMP) improved this to $4n$
states, independently of the input alphabet~\cite{GMP}. They asked
whether complementation can require any increase at all: namely,
whether there exists an $n$-state 2DFA whose complement cannot be
recognized by an $n$-state 2DFA. This is the fifth open problem in
their concluding section~\cite[Section~6, p.~1186, problem~(e)]{GMP}.

We answer this question affirmatively under the acceptance convention
fixed below. Here $\statecomplexity(L)$ denotes the \emph{state complexity}
of $L$: the minimum number of states of a 2DFA recognizing $L$ with
acceptance at the left endmarker and infinite computations counted as
rejection. For each $n\ge22$, we construct an $n$-state 2DFA over
the fixed alphabet $\Sigma=\{\mathtt A,\mathtt B,\mathtt C\}$ whose
language $K_n$ satisfies
\[
 \statecomplexity(K_n)=n,\qquad \statecomplexity(\Sigma^*\setminus K_n)=2n.
\]
Both lower bounds apply to arbitrary deterministic recognizers in this
model.
The same family also answers GMP's question in their original model,
where a final state accepts at any tape position: its exact source
and complement complexities are $N$ and $2N-1$, for every $N\ge23$.
Section~\ref{sec:conventions} proves this comparison. The leading
factor is two in both conventions.

The witness divides its states into two sets, or \emph{banks}, one
for forward scans and one for backward scans. Each input word induces
a transformation on each bank. A forward scan followed by a backward
scan tests an acceptance condition. If the first test fails, the
machine makes a second round from a different starting state. If this
test also fails, the same round repeats forever. Thus every rejecting
computation of the witness is infinite.

We then consider an arbitrary automaton for $\overline{K_n}$. Its
behavior on a nonempty input interval is summarized by the state and
side of the first exit for every possible entry. We call this record
an \emph{interval diagram}. Equal diagrams can be substituted in any
fixed context without changing acceptance. The witness supplies
contexts that distinguish independent permutations on the two banks.
These distinctions force a group of target diagrams to map onto a
product of two alternating groups. This group acts on the states in
which runs cross the interval, with the two exit directions counted
separately. We call these signed exits \emph{through channels}.
Determinism bounds their number by the number of target states.

If the complementing machine had fewer than $2n$ states, we show that
one alternating factor would act naturally on just one orbit of
channels while fixing all other channels. On this subgroup, each
traversal through the moving orbit amounts to a query to a permutation
or its inverse. An answer leaves the machine in a fixed state at a
fixed boundary, so the next action depends only on that answer and
cannot also depend on which query produced it. We prove that such a
controller cannot perform the test required by complementation. This
forces at least $2n$ states. The same group argument also proves that
the witness language itself requires $n$ states.

A matching upper bound follows from reverse traversal of the
configuration graph. In the witness, every transition entering a
given state moves in the same direction. For sources with this
property, called \emph{direction-determinate}, the traversal uses
$2n$ states. We give its full transition table for the present model.
For an arbitrary source, splitting the $b$ states entered in both
directions gives a halting complement with at most $2n+2b$ states.
Thus a lower bound with a leading constant greater than two would
require linearly many source states entered in both directions.

\paragraph{Related work.}
State costs of simulations between one-way and two-way automata were
studied by Sakoda and Sipser~\cite{SakodaSipser1978};
Pighizzini~\cite{Pighizzini2012} surveys subsequent work.
Jir\'askov\'a and Okhotin~\cite{JiraskovaOkhotin2017} study operations
on two-way automata, and Kunc and Okhotin~\cite{KuncOkhotin2012}
treat unary alphabets. Kapoutsis~\cite{Kapoutsis2018} studies one-way
liveness on inputs of length two and three. Our alphabet has three
letters, with no restriction on input length. The reverse traversal
for direction-determinate automata belongs to the graph-walking
constructions of Kunc and Okhotin~\cite{kuncokhotin2020}; its $2n$
state bound is restated by Martynova and
Okhotin~\cite[Lemma~8]{martynovaokhotin2021}.

For background on semigroups and permutation groups, see
Pin~\cite{Pin1997}, Howie~\cite{Howie1995}, and Dixon and
Mortimer~\cite{DixonMortimer1996}. Our small-index argument uses
Shareshian and Woodroofe~\cite[Lemma~4.2]{SW}. Related descriptions
of interval behavior occur in OpenAI's preprints on deterministic
simulation of one-way liveness~\cite{OpenAILiveness2026} and
nondeterministic complementation~\cite{OpenAIComplementation2026}.
Those papers concern different state-complexity questions. The proof
here gives the interval-diagram and group arguments needed for this
lower bound, apart from the cited standard group-theoretic facts.
No unproved conjecture is assumed.

\section{Preliminaries and main result}\label{sec:model}

\subsection{Automata and acceptance}

A 2DFA is a five-tuple $A=(Q,\Sigma,\delta,s,f)$. Here $Q$ is a
nonempty finite set of states, $\Sigma$ is a finite ordinary alphabet,
$s,f\in Q$, and
\[
 \delta:Q\times(\Sigma\sqcup\{\LL,\RR\})
       \rightharpoonup Q\times\{-1,+1\}
\]
is a partial function. The endmarkers $\LL,\RR$ are not in $\Sigma$.
Defined transitions on $\LL$ move right, and defined transitions on
$\RR$ move left. On a word $w=x_1\cdots x_t$, the input occupies
positions $1,\ldots,t$, with endmarkers at $0$ and $t+1$. A
configuration is a pair $(q,i)$ of a state and a tape position.

The computation begins in $(s,0)$ and accepts exactly upon reaching
$(f,0)$. Acceptance takes priority over any outgoing transition there.
An undefined transition at any other configuration halts and rejects;
an infinite computation also rejects. In particular, state $f$ is not
accepting at an ordinary cell or at the right endmarker. There are no
stationary moves or uncounted halting controls. Every member of $Q$,
including $s$ and $f$, contributes to the state count.

\subsection{State complexity}

Write $\statecomplexity(L)$ for the minimum number of states of a 2DFA
recognizing $L$ in this model. Fix
$\Sigma=\{\mathtt A,\mathtt B,\mathtt C\}$, write
$\overline L=\Sigma^*\setminus L$, and, for $n\ge1$, set
\begin{equation}\label{eq:cost}
 C_3(n)=\max\{\statecomplexity(\overline{L(A)}):
                 A\text{ is over }\Sigma,\ |Q_A|\le n\}.
\end{equation}
The maximum exists: there are finitely many transition tables up to
state renaming, and the recognized languages are regular. Their
complements can also be realized in this model by a rightward
finite-state scan followed, upon success, by an accepting leftward
sweep. The bound in~\eqref{eq:cost} is on the number of states of the given
source. We will also prove that no smaller source recognizes our
witness language.

\begin{theorem}[Exact factor-two witness]\label{thm:main}
For every $n\ge22$, there is an explicitly specified $n$-state 2DFA
over $\Sigma$ whose language $K_n$ satisfies
\[
 \statecomplexity(K_n)=n,\qquad \statecomplexity(\overline{K_n})=2n.
\]
Consequently $C_3(n)\ge2n$ for every $n\ge22$. Both lower bounds
range over arbitrary deterministic recognizers in the stated model,
including recognizers with infinite rejecting computations.
\end{theorem}

Put $k=\lfloor n/2\rfloor$ and $\ell=\lceil n/2\rceil$. Then
$k\ge11$, $\ell\in\{k,k+1\}$, and $n=k+\ell$.
We construct $K_n=L_{k,\ell}$ below. Functions and permutations are
composed from right to left. For a positive integer $t$, write
$[t]=\{1,\ldots,t\}$, and let $\Sym_t$ and $\Alt_t$ be the
symmetric and alternating groups on $[t]$.

\section{The ternary witness}
\label{sec:new-witness}

\subsection{The source five-tuple and its language}

Fix integers $k\ge 11$ and $\ell\in\{k,k+1\}$. For every $t\ge 11$, define
\[
 c_t=\begin{cases}(1\ 2\ \cdots\ t),&t\text{ odd},\\
                  (1\ 2\ \cdots\ t-1),&t\text{ even},\end{cases}
 \qquad
 d_t=\begin{cases}(1\ 2),&t\text{ odd},\\
                  (t-1\ t),&t\text{ even},\end{cases}
\]
and let $e_t$ fix every point except $1$, which it sends to $2$.
The order $r_t$ of $c_t$ is odd. Moreover, $c_t,d_t$ generate $\Sym_t$:
in odd degree these are a full cycle and an adjacent transposition; in even
degree the conjugates of $d_t$ by powers of $c_t$ are all transpositions
$(i\ t)$, $1\le i<t$.

Use the fixed ordinary alphabet $\Sigma=\{\mathtt A,\mathtt B,\mathtt C\}$.
Associate to each letter two transformations as follows:
\[
\begin{array}{c|cc}
 x&\varphi_x:[k]\to[k]&\psi_x:[\ell]\to[\ell]\\ \hline
 \mathtt A&c_k&d_\ell\\
 \mathtt B&d_k&c_\ell^{-1}\\
 \mathtt C&e_k&e_\ell
\end{array}
\]
The source is the following five-tuple, with exactly $k+\ell$ states:
\begin{equation}\label{eq:source-tuple}
Q=\{P_i:i\in[k]\}\sqcup\{Q_j:j\in[\ell]\},
 \qquad A_{k,\ell}=(Q,\Sigma,\delta,P_1,Q_1).
\end{equation}
For every ordinary $x\in\Sigma$, its transitions are
\begin{equation}\label{eq:source-ordinary}
\delta(P_i,x)=(P_{\varphi_x(i)},+1),\qquad
 \delta(Q_j,x)=(Q_{\psi_x(j)},-1).
\end{equation}
The defined marker rows are
\begin{equation}\label{eq:source-markers}
\begin{array}{c|c}
 \text{row}&\text{transition}\\ \hline
 (P_1,\LL)&(P_1,+1)\\
 (P_1,\RR)&(Q_1,-1)\\
 (P_i,\RR),\ i\ne1&(Q_2,-1)\\
 (Q_j,\LL),\ j\ne1&(P_2,+1).
\end{array}
\end{equation}
All remaining marker rows are undefined. As stipulated by the model,
$(Q_1,\LL)$ accepts immediately.

For $w=x_1\cdots x_t$, put
\begin{equation}\label{eq:word-maps}
f_w=\varphi_{x_t}\circ\cdots\circ \varphi_{x_1},\qquad
 g_w=\psi_{x_1}\circ\cdots\circ \psi_{x_t},
 \qquad f_\varepsilon=g_\varepsilon=\mathrm{id}.
\end{equation}
The different orders reflect the two scan directions. In particular,
$f_{uv}=f_v\circ f_u$ and $g_{uv}=g_u\circ g_v$.
Define $\chi:[k]\to\{1,2\}$ by $\chi(1)=1$ and $\chi(i)=2$ for $i\ne1$.

\begin{lemma}[Witness language]\label{lem:witness-language}
The displayed source recognizes
\begin{equation}\label{eq:witness-language}
L_{k,\ell}=\bigl\{w:
 g_w(\chi(f_w(1)))=1\ \text{or}\ g_w(\chi(f_w(2)))=1\bigr\}.
\end{equation}
Every word outside this language produces an infinite run.
\end{lemma}
\begin{proof}
The initial forward scan starts in $P_1$ and reaches the right marker in
$P_{f_w(1)}$. Its marker row starts the reverse scan in $Q_{\chi(f_w(1))}$.
That scan reaches the left marker in $Q_{g_w(\chi(f_w(1)))}$, accepting exactly
when its index is $1$. Otherwise, its marker row starts a forward scan in
$P_2$. Every subsequent round tests $g_w(\chi(f_w(2)))=1$ in exactly the same
way. If this test fails, the configuration immediately after the next
left-marker transition repeats. The empty word accepts on its first
round, consistently with the formula.
\end{proof}

Figure~\ref{fig:witness} summarizes the source computation.

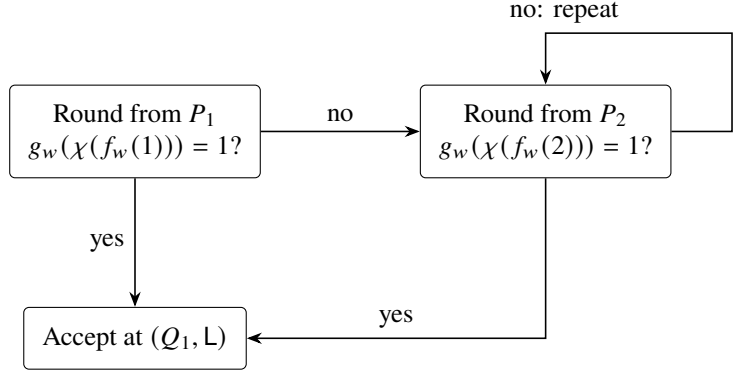
\begin{figure}[tb]
\centering
\begin{tikzpicture}[>=Stealth,node distance=17mm and 21mm,
 every node/.style={font=\small},
 phase/.style={draw,rounded corners=2pt,align=center,inner sep=7pt},
 flow/.style={->,semithick}]
\node[phase] (one) {Round from $P_1$\\$g_w(\chi(f_w(1)))=1$?};
\node[phase,right=of one] (two) {Round from $P_2$\\$g_w(\chi(f_w(2)))=1$?};
\node[phase,below=of one] (accept) {Accept at $(Q_1,\LL)$};
\draw[flow] (one) -- node[above]{no} (two);
\draw[flow] (one) -- node[left]{yes} (accept);
\draw[flow] (two.south) |- node[pos=.75,above]{yes} (accept.east);
\draw[flow] (two.east) -- ++(8mm,0) -- ++(0,13mm)
 -| node[pos=.45,above]{no: repeat} (two.north);
\end{tikzpicture}
\caption{Each round consists of a forward scan in the $P$-states and
a backward scan in the $Q$-states. All later rounds begin in $P_2$.
The boxes describe phases of the displayed automaton, with no
additional states.}
\label{fig:witness}
\end{figure}

\subsection{Words for independent transformations}

For a word over $\{\mathtt A,\mathtt B\}$, both transformations are
permutations; write $h_w=g_w^{-1}$. Then
\[
 (f_{uv},h_{uv})=(f_v\circ f_u,h_v\circ h_u).
\]
Thus $\mathtt A$ and $\mathtt B$ induce the pairs $(c_k,d_\ell)$
and $(d_k,c_\ell)$ in the coordinates $(f_w,h_w)$. Since $r_k,r_\ell$ are odd,
\[
\begin{array}{c|c}
 \text{word}&(f_w,h_w)\\ \hline
 \mathtt A^{r_k}&(\mathrm{id},d_\ell)\\
 \mathtt A^{r_k+1}&(c_k,\mathrm{id})\\
 \mathtt B^{r_\ell}&(d_k,\mathrm{id})\\
 \mathtt B^{r_\ell+1}&(\mathrm{id},c_\ell).
\end{array}
\]
Consequently these ordinary words generate $\Sym_k\times\Sym_\ell$.
Equivalently, every independent pair of permutations $(f_w,g_w)$ has a
word representative. Representatives may always be made nonempty by
inserting the identity word $\mathtt A^{2r_k}$.

\begin{lemma}[Paired singular generation]\label{lem:paired-generation}
For every pair of singular (nonbijective) transformations
$f:[k]\to[k]$ and $g:[\ell]\to[\ell]$, there is an ordinary word $w$
with $(f_w,g_w)=(f,g)$.
\end{lemma}
\begin{proof}
First, independent permutation words conjugate $\mathtt C$ into a word
representing any prescribed pair of elementary collapses. More explicitly,
if $u,v$ have transformation pairs $(p,q)$ and $(p^{-1},q^{-1})$, then
\[
 (f_{u\mathtt Cv},g_{u\mathtt Cv})
   =(p^{-1}e_kp,\,q e_\ell q^{-1}).
\]
The permutations $p,q$ can be chosen independently.

Every singular transformation of $[t]$ is a permutation composed after a
nonempty product of elementary collapses. Indeed, choose one representative
of each nonempty fiber and collapse each other point onto that representative.
A permutation then maps the representatives to their prescribed images.
Every elementary collapse is idempotent, so repeating a factor
consecutively lengthens the product without changing its value.
We may therefore write, with the same positive integer $r$,
\[
 f=\pi\circ a_r\circ\cdots\circ a_1,
 \qquad
 g=\rho\circ b_r\circ\cdots\circ b_1,
\]
where all $a_i,b_i$ are elementary collapses. Choose an initial permutation word
with pair $(\mathrm{id},\rho)$, then $r$ words with pairs
$(a_i,b_{r+1-i})$ in the order $i=1,\ldots,r$, and a final permutation word with
pair $(\pi,\mathrm{id})$. Their concatenation has pair $(f,g)$: the first
coordinate composes in reverse physical order, while the second composes
in physical order, as required for the backward scan.
\end{proof}

The lemma covers pairs in which both coordinates are singular.
All contexts used below have this form. In particular, define the
singular map $\eta_t$ by
\[
 \eta_t(1)=1,\qquad \eta_t(j)=2\quad(j\ne1).
\]
It is idempotent and fixes both $1$ and $2$. When an inactive bank need
only preserve these two distinguished points, $\eta_t$ supplies a
singular coordinate for the paired-generation lemma.

\section{Interval diagrams and through channels}\label{sec:diagrams}

\subsection{Composition and substitution}

Fix an arbitrary target $B$ with $m$ states recognizing
$L_{k,\ell}$ or $\overline{L_{k,\ell}}$. No structural assumption is made about its
transition table. In particular, it may change direction on every
ordinary symbol and may reject by looping.

For a nonempty ordinary word $v$, run $B$ on just its interval of
cells, stopping at the first step outside that interval. A left entry
in state $q$ means that the head is on its first cell in state $q$;
a right entry means that it is on its last cell in state $q$. There
are four partial exit functions:
\begin{equation}\label{eq:exit-table}
 \begin{array}{c|cc}
       &\text{exit left}&\text{exit right}\\\hline
 \text{enter left}&L_v&F_v\\
 \text{enter right}&B_v&R_v.
 \end{array}
\end{equation}
The output is the state in which the head first exits. For a given
entry there is at most one finite exit. If the run hits an undefined
transition or remains inside forever, both applicable functions are
undefined. These two failures may be identified because both reject
and neither can interact further with the outside of the interval.
There is no accepting configuration inside this interval: it contains
only ordinary letters, whereas acceptance is confined to the actual
left endmarker.

The complete interval diagram $D_v=(F_v,B_v,L_v,R_v)$ records these
functions on all states and both entry sides. An entry or exit, also called a \emph{port}, is specified by a state
and a side. Ports on opposite sides are distinct even when they have
the same state label.

\begin{lemma}[Composition and substitution]\label{lem:diagrams}
The diagram $D_{uv}$ is determined by $D_u,D_v$. Equality of diagrams
permits substitution between any fixed ordinary left and right
contexts with the actual endmarkers and initial/final convention.
There are finitely many diagrams for an $m$-state target.
\end{lemma}
\begin{proof}
At the join of $u$ and $v$, an exit to the right of $u$ becomes
a left entry into $v$ with the same state, and a left exit of $v$
becomes a right entry into $u$ with the same state. Follow these
instructions until an exterior exit or a failure occurs. If joins
are crossed infinitely often, there is no exterior exit. Equivalently,
the finite join-control graph eventually repeats a control, which
also yields no exterior exit. This computes all four functions of
$D_{uv}$ from the two diagrams. It describes the actual run, so the
composition is associative.

The same calculation with arbitrary fixed outside regions preserves
the whole deterministic run at every visit to an interval boundary.
A run terminating or accepting in an outside region is therefore
preserved. A run failing inside, or making infinitely many boundary
visits without accepting, remains rejecting. This proves substitution,
including loops and the actual accepting left endmarker configuration. Finally,
each of $2m$ entries has at most $2m$ exits or failure, so at most
$(2m+1)^{2m}$ possible complete exit tables exist. Some such tables
need not be realizable, but this upper bound proves finiteness.
\end{proof}

It is useful to specify the multiplication convention explicitly.
Write $D_u\diamond D_v=D_{uv}$ for physical concatenation. For
group arguments use instead
\begin{equation}\label{eq:product}
                         D_uD_v=D_v\diamond D_u=D_{vu}.
\end{equation}
This reverses the order of concatenation so that
$D_w\mapsto(f_w,g_w^{-1})$ respects ordinary composition of functions.
The convention concerns only the notation for multiplication.

\subsection{The quotient determined by the source transformations}

For the fixed recognizer, let $M$ consist of the
diagrams of nonempty $\{\mathtt A,\mathtt B\}$-words whose semantic pair
$(f_w,g_w^{-1})$ belongs to $\Alt_k\times\Alt_\ell$. We call this
the \emph{semantic pair}: it records the source transformations,
whereas $D_w$ records the target behavior. The set $M$ is a finite
semigroup under~\eqref{eq:product}. Write
$G=\Alt_k\times\Alt_\ell$. Every element of $G$ has an allowed
word representative by the independent generation above.

\begin{lemma}[Product semantic quotient]\label{lem:semantic-product}
The assignment
\[
 \theta(D_w)=(f_w,g_w^{-1})
\]
is a well-defined surjective homomorphism
$M\twoheadrightarrow\Alt_k\times\Alt_\ell$.
\end{lemma}
\begin{proof}
We prove that distinct semantic pairs are separated by ordinary contexts.
For a middle word with pair $(\pi,\sigma)$, put $g=\sigma^{-1}$. For
contexts $\lambda,\rho$, the complete transformations are
\[
 f_{\lambda w\rho}=f_\rho\circ\pi\circ f_\lambda,
 \qquad
 g_{\lambda w\rho}=g_\lambda\circ g\circ g_\rho.
\]
Consider two pairs $(\pi,\sigma)$ and $(\pi',\sigma')$, writing
$g'=\sigma'^{-1}$.

If $g\ne g'$, choose $c$ with $g(c)\ne g'(c)$. Take $f_\lambda$ and
$f_\rho$ constantly $1$, take $g_\rho$ constantly $c$, and let $g_\lambda$
map $g(c)$ to $1$ and every other point to $2$. Each context has two
singular coordinates and therefore exists by
Lemma~\ref{lem:paired-generation}. The complete forward transformation is
constant $1$, and the two reverse outcomes are respectively $1$ and $2$.
Lemma~\ref{lem:witness-language} distinguishes the two words.

If $g=g'$ but $\pi\ne\pi'$, choose $a$ with $\pi(a)\ne\pi'(a)$.
Take $f_\lambda$ constantly $a$, and let $f_\rho$ map $\pi(a)$ to $1$
and every other point to $2$. Take $g_\rho=\eta_\ell$, and let
$g_\lambda$ map $g(1)$ to $1$ and every other point to $2$.
These are again paired singular contexts. Since $g$ is a permutation,
$g(1)\ne g(2)$. For the first middle word both forward tests launch the
reverse scan at $1$ and accept; for the second both launch it at $2$ and
fail. The source answers, and hence the target answers, differ.

Equality of interval diagrams permits substitution in fixed ordinary
contexts, so these separations prove well-definedness. Surjectivity follows
from independent permutation generation. Finally,
\[
 \theta(D_uD_v)=\theta(D_{vu})
 =\bigl(f_u\circ f_v,\,g_u^{-1}\circ g_v^{-1}\bigr)
 =\theta(D_u)\theta(D_v),
\]
which proves the homomorphism law.
\end{proof}

\subsection{Extracting a group}

For an idempotent $e$ (that is, $e^2=e$), the set
$eMe=\{exe:x\in M\}$ is a monoid with identity $e$, called its
\emph{local monoid}. A \emph{unit} is an element invertible in that
monoid. The following elementary argument extracts a local monoid
that is a group.

\begin{lemma}\label{lem:group-extraction}
There is an idempotent $e\in M$ such that $eMe$ is a group mapping
onto $G$. Its identity and every element have nonempty actual
word representatives.
\end{lemma}
\begin{proof}
Every element of a finite semigroup has a positive idempotent power.
For completeness, its powers eventually have a period $d$; choose
a multiple $t$ of $d$ large enough that $t$ and $2t$ are in the
periodic part. Then $x^{2t}=x^t$.

On idempotents use the partial order $p\le e$ when $pe=ep=p$.
Reflexivity and antisymmetry are immediate; transitivity follows by
substitution in these two identities. Choose a minimal idempotent
$e$. For $h\in eMe$, an idempotent positive power $h^t$ satisfies
$eh^t=h^te=h^t$, so $h^t\le e$ and hence $h^t=e$. Consequently
every $h$ is a unit with inverse $h^{t-1}$ when $t>1$, and inverse
$e$ when $t=1$. Thus $eMe$ is a group with identity $e$.

The only idempotent in $G$ is its identity, so $\theta(e)=1$.
For any $x\in M$, the element $exe$ has image $\theta(x)$.
Thus $eMe$ still maps onto $G$. Every element belongs to $M$ and
is therefore represented by a nonempty word.
\end{proof}

Fix this idempotent diagram $e$. Write
\begin{equation}\label{eq:channels}
                    U=\im F_e,\qquad V=\im B_e.
\end{equation}
These are signed sets: $U$ consists of right-exit labels and $V$
of left-exit labels. They are treated as disjoint even when they
contain the same underlying state name.

\subsection{The action on through channels}

\begin{lemma}[Unit channels]\label{lem:unit-channels}
Every unit $g$ of $eMe$ has the same return functions and through
domains as $e$. There are unique permutations $\alpha_g$ of $U$
and $\beta_g$ of $V$ such that
\begin{equation}\label{eq:unit-form}
 L_g=L_e,\quad R_g=R_e,\quad
 F_g=\alpha_g\circ F_e,\quad B_g=\beta_g\circ B_e.
\end{equation}
The action $g\mapsto(\alpha_g,\beta_g^{-1})$ is faithful and,
under multiplication \eqref{eq:product}, is a homomorphism into
$\Sym(U)\times\Sym(V)$.
\end{lemma}
\begin{proof}
Let $h$ be the inverse of $g$. A finite left return of the first
factor in a physical concatenation persists unchanged in the product;
a right return of the last factor likewise persists. Since
$g=e\diamond g\diamond e$, every left/right return of $e$ is a
return of $g$. Conversely, applying these observations to
$g\diamond h=e$ and $h\diamond g=e$ gives every return of $g$
as a return of $e$. Hence $L_g=L_e$ and $R_g=R_e$.

A failed exterior entry of the outer $e$ in $e\diamond g\diamond e$
cannot become a successful entry of $g$. Conversely, if a defined
left entry of $e$ failed in $g$, it would fail in $g\diamond h=e$;
the corresponding right-entry assertion uses $h\diamond g=e$.
Thus the domains of successful left and right entries agree. After
removing the equal return domains, the forward and backward through
domains agree as well.

Every final right exit of $e\diamond g\diamond e$ uses a forward
row of its last $e$, so $\im F_g\subseteq U$. If two forward
entries have equal $F_e$ outputs, their paths merge after the first
factor in $e\diamond g=g$, giving equal $F_g$ outputs. Conversely,
equal $F_g$ outputs merge in $g\diamond h=e$, giving equal $F_e$
outputs. The functions $F_g,F_e$ therefore have exactly the same
fibers on their common domain, and hence the same number of image
points. The image inclusion is consequently equality, and there
is a unique permutation $\alpha_g$ satisfying \eqref{eq:unit-form}. Reflection gives
the statement for $B_g$ and $\beta_g$.

To prove the multiplication law, we examine what happens at the
join between two intervals. In $e\diamond e=e$, enter the
join with a right-exit label $u\in U$ from the left factor. Such
a label is produced by some forward entry of $e$. At the join,
successive left returns of the right factor and right returns of
the left factor use only $L_e,R_e$. This fixed chain must eventually
reach a forward-domain entry $r_+(u)$ in the right factor. It cannot
fail, loop, or exit to the exterior left, because the corresponding
forward row of $e\diamond e$ equals that of $e$. Its ultimate
output satisfies $F_e(r_+(u))=u$.

For units $g,h$, the return rows and through domains are identical,
so the join chain after a forward output $u$ of $g$ is exactly this
same fixed chain. The output through $h$ is now $\alpha_h(u)$.
It follows that, in physical order,
$\alpha_{g\diamond h}=\alpha_h\alpha_g$. Reflection gives
$\beta_{g\diamond h}=\beta_g\beta_h$. Under \eqref{eq:product} these laws become
$\alpha_{gh}=\alpha_g\alpha_h$ and
$\beta_{gh}=\beta_h\beta_g$, proving the stated homomorphism.
Finally \eqref{eq:unit-form}, together with the fixed return and failure rows,
determines the whole diagram. Its representation by the pair of
permutations is therefore faithful. In particular,
$\alpha_e=\mathrm{id}_U$ and $\beta_e=\mathrm{id}_V$.
\end{proof}

\begin{lemma}[The number of through channels]\label{lem:capacity}
The signed through-image sets satisfy
\begin{equation}\label{eq:capacity}
                              |U|+|V|\le m.
\end{equation}
\end{lemma}
\begin{proof}
Choose an actual nonempty word representing $e$ and one through run
for each distinct signed exit in $U\sqcup V$. Each chosen run
visits every cell of that word. No two can contain the same state
on the same cell, because the two deterministic suffix runs would
then be identical and would give the same signed exit. Fix any one
ordinary cell. The chosen runs consequently use distinct physical
states there, giving \eqref{eq:capacity}. The argument uses one
ordinary cell and therefore requires a nonempty word; an empty
interval would formally have $2m$ through ports.
\end{proof}

\section{Small actions of a product of alternating groups}\label{sec:groups}

We next show why fewer than $2(k+\ell)$ channels force one alternating
factor to act on just one natural orbit. The main step is
Lemma~\ref{lem:isolation}. Its proof first limits the possible orbit
sizes, then removes the actions on the other orbits by taking derived
subgroups. The first two lemmas provide the required index bounds.

An action is \emph{faithful} if only the identity fixes every point.
A transitive permutation group is \emph{imprimitive} if it preserves
a partition into equal-sized blocks with sizes strictly between one
and the degree; otherwise it is \emph{primitive}. The derived series
of a group $K$ is defined by $K^{(0)}=K$ and
$K^{(j+1)}=[K^{(j)},K^{(j)}]$. A group is \emph{solvable} if some
term is trivial, and \emph{perfect} if $K^{(1)}=K$. We use the
standard facts that $\Alt_d$ is nonabelian simple and perfect for
$d\ge5$, and that every subgroup of $\Sym_r$ for $r\le4$ has
derived length at most three; see~\cite{DixonMortimer1996}.

We use the following theorem of Shareshian and Woodroofe:
for $d\ge9$, every prime $p\le d-3$ divides the index of every
proper primitive subgroup of $\Alt_d$. This is
\cite[Lemma~4.2, p.~230]{SW}. For odd $p$, its proof uses Sylow
conjugacy and Jordan's prime-cycle theorem. For $p=2$, it uses the
double-transposition part of Jordan's theorem, stated as
\cite[Theorem~4.1(2)]{SW}. Together with Bertrand's postulate, this
implies the elementary small-index bounds below.

First recall that every proper subgroup of $\Alt_d$, for $d\ge5$, has
index at least $d$. Indeed, its nontrivial coset action is faithful by
simplicity. An index $r<d$ would therefore give
$d!/2\le r!\le(d-1)!$, a contradiction.

\begin{lemma}[Subgroups of small index]\label{lem:small-index}
If $t\ge11$ and $J<\Alt_t$ satisfies $[\Alt_t:J]<4t+2$, then $J$ is the
stabilizer of a point in the natural action on $[t]$.
\end{lemma}
\begin{proof}
Suppose first that $J$ is primitive. For $11\le t\le13$, the
primitive-subgroup fact forces the distinct primes $2,3,5,7$ to divide
its index. Their product $210$ exceeds $4t+2$. For $t\ge14$, put
$M=\lfloor(t-3)/2\rfloor\ge5$. Bertrand's postulate supplies a prime
$M<p<2M\le t-3$. This prime is at least $7$, so the distinct primes
$2,3,5,p$ divide the index. Consequently
\[
 [\Alt_t:J]\ge30p>30M\ge15(t-4)>4t+2,
\]
again a contradiction.

If $J$ is transitive and imprimitive, it preserves a partition into
blocks of a common size $b$, where $2\le b\le t/2$ and $b$ divides $t$.
The group $\Alt_t$ is transitive on such partitions: an odd transporter
can be made even by a transposition within a block of the image
partition. Thus the index of the partition stabilizer is the number of
these partitions, and is a lower bound for $[\Alt_t:J]$. For $b\ge3$,
choosing the other $b-1$ members of the block containing one specified
point gives at least
\[
 \binom{t-1}{b-1}\ge\binom{t-1}{2}>4t+2
\]
partitions for $t\ge12$. For $b=2$ the number is
$(t-1)!!\ge(t-1)(t-3)>4t+2$, also for $t\ge12$.
There is no imprimitive case when $t=11$, since $11$ is prime.

Finally suppose that $J$ is intransitive. Unless it has an orbit of
size $t-1$, a union of its orbits has some size $a$ with
$2\le a\le t-2$. Its setwise stabilizer has index $\binom ta$ in
$\Alt_t$, whence
\[
 [\Alt_t:J]\ge\binom ta\ge\binom t2>4t+2.
\]
Here transitivity on $a$-subsets follows by correcting an odd
transporter with a transposition inside its image subset. The remaining
case fixes one point and has $J\le\Alt_{t-1}$. If this inclusion were
proper, the minimum-index observation above would give
$[\Alt_t:J]\ge t(t-1)>4t+2$. Therefore $J$ is the full point stabilizer.
\end{proof}

\begin{lemma}[Small-index product subgroups]\label{lem:product-index}
Let $k\ge11$, $\ell\in\{k,k+1\}$, and $G=\Alt_k\times\Alt_\ell$.
Every proper subgroup $J<G$ with $[G:J]<2(k+\ell)$ has one of the forms
\[
 J=P\times\Alt_\ell
 \qquad\text{or}\qquad
 J=\Alt_k\times Q,
\]
where $P$ or $Q$, respectively, is a natural point stabilizer.
\end{lemma}
\begin{proof}
Let $P\le\Alt_k$ and $Q\le\Alt_\ell$ be the projections of $J$.
If both are proper, then
\[
 [G:J]\ge[\Alt_k:P][\Alt_\ell:Q]\ge k\ell>2(k+\ell),
\]
which is impossible. Suppose, without loss of generality, that
$Q=\Alt_\ell$. The subgroup
$J\cap(\{1\}\times\Alt_\ell)$ is normal in $\{1\}\times\Alt_\ell$,
because the second projection of $J$ is surjective. Simplicity makes
this intersection trivial or the whole second factor. In the former
case the first projection embeds $J$ into $P$, and hence
\[
 [G:J]=\frac{|\Alt_k||\Alt_\ell|}{|J|}
 \ge |\Alt_\ell|>2(k+\ell),
\]
a contradiction. Thus $J=P\times\Alt_\ell$, with $P$ proper since
$J$ is proper. Now $[\Alt_k:P]<2(k+\ell)\le4k+2$, so
Lemma~\ref{lem:small-index} applies. The other case is symmetric,
using $2(k+\ell)\le4\ell+2$.
\end{proof}

\begin{lemma}[Isolation of one natural factor orbit]\label{lem:isolation}
Let $k\ge11$, $\ell\in\{k,k+1\}$, and $G=\Alt_k\times\Alt_\ell$.
Suppose a finite group $H$ acts faithfully on a finite set $X$ and
admits a surjective homomorphism $\theta:H\twoheadrightarrow G$.
If $|X|<2(k+\ell)$, then $H$ has a subgroup $\Gamma$ with the following
properties. For one $d\in\{k,\ell\}$, the restriction of $\theta$ is
an isomorphism from $\Gamma$ onto the corresponding factor $\Alt_d$
(with the other factor equal to the identity), and $\Gamma$ has exactly
one nonfixed orbit on $X$. This orbit has size $d$ and is the natural
action of that same semantic factor. Every other point of $X$ is fixed.
If $X=U\sqcup V$ with both parts $H$-invariant, this orbit lies wholly
in $U$ or wholly in $V$.
\end{lemma}
\begin{proof}
Choose an inclusion-minimal subgroup $\Lambda\le H$ still surjecting
onto $G$, and write $N=\ker(\theta|_\Lambda)$. The restricted action
of $\Lambda$ remains faithful. For a nonfixed point $x$, let
$H_x=\Lambda_x$ and $J_x=\theta(H_x)$. Since $H_x$ is proper,
minimality makes $J_x$ proper as well. The index formula gives
\begin{equation}\label{eq:product-orbit-fibers}
 |\Lambda x|=[\Lambda:H_x]
   =[G:J_x]\,[N:N\cap H_x].
\end{equation}
Lemma~\ref{lem:product-index} therefore assigns each nonfixed
$\Lambda$-orbit to one of the two factors: its size is $dr$, where
$d$ is the degree of that factor and
$r=[N:N\cap H_x]\ge1$. Since
$dr<2(k+\ell)\le4d+2$ and $d\ge11$, necessarily $r\le4$.

More precisely, the $N$-orbits within $\Lambda x$ form $d$ blocks,
each of size $r$. Indeed, $NH_x=\theta^{-1}(J_x)$, so these blocks
are indexed by $\Lambda/(NH_x)$. Their induced $\Lambda$-action
is the coset action of $G$ on $G/J_x$, which is the natural action of the assigned factor.
The other factor of $G$ fixes each of these blocks.

Each factor must have at least one orbit assigned to it. Otherwise,
let $K$ be the full preimage in $\Lambda$ of that factor, with the
other factor equal to the identity. On every nonfixed orbit, $K$
would fix every $N$-block setwise. Its action would consequently embed
into a product of symmetric groups on sets of size at most four;
globally fixed points contribute trivial factors. Faithfulness would
make $K$ solvable, since each $\Sym_r$ for $r\le4$ is solvable. This is
impossible because $K$ surjects onto the nonabelian simple alternating
factor.

Label the two factors by $1,2$, with degrees $d_1=k$ and $d_2=\ell$,
even when $k=\ell$. Let $a_1,a_2$ be the total numbers of points
in the orbits assigned to them. We have
$a_1+a_2\le|X|<2d_1+2d_2$. Thus for one degree $d$ its assigned
points number less than $2d$. This positive number is a sum of positive
integer multiples of $d$, so it is exactly $d$. There is exactly one
orbit assigned to this factor, and it has $r=1$.

Let $K$ again denote the full preimage of this factor in $\Lambda$.
On every orbit assigned to the other factor, its image lies in a
direct product of groups $\Sym_r$ with $r\le4$: the blocks are
fixed setwise, so there is no additional permutation of blocks. Such
an image has derived
length at most three: for $\Sym_4$ the successive derived subgroups
are $\Alt_4$, the Klein four group, and the identity. Therefore
$\Gamma=K^{(3)}$, the third derived subgroup of $K$, fixes all these
other-factor orbits. It also fixes the globally fixed points.

Because $\Alt_d$ is perfect, $\theta(\Gamma)=\Alt_d$. On the unique
orbit assigned to this factor the fibers have size one, so $N$ acts
trivially and the action is the natural action of the same semantic
$\Alt_d$. Explicitly, if $J_x$ fixes point $a$ in the selected
factor $i$, label $\lambda x$ by $\theta_i(\lambda)(a)$. The
stabilizer identity makes this labeling well defined and bijective,
and $\gamma\in\Gamma$ acts on the labels by
$\theta_i(\gamma)$. Any element of $\Gamma\cap N$ therefore fixes all of $X$;
faithfulness implies that it is the identity. Thus
$\theta|_\Gamma$ is the required isomorphism. Finally any
$\Gamma$-orbit is contained in an invariant part $U$ or $V$, proving
the signed-orbit assertion.
\end{proof}

\begin{corollary}[Minimum degree of the product action]\label{cor:product-capacity}
Under the group and degree hypotheses of
Lemma~\ref{lem:isolation}, every faithful action of a finite group
surjecting onto $\Alt_k\times\Alt_\ell$ has at least $k+\ell$ points.
In particular, if a group of actual interval diagrams surjects onto
this product, its faithful signed through-channel action satisfies
$|U|+|V|\ge k+\ell$.
\end{corollary}
\begin{proof}
An action with fewer than $k+\ell$ points is below the threshold
$2(k+\ell)$. Apply the orbit-and-block analysis in the proof of
Lemma~\ref{lem:isolation}: each factor must have an assigned orbit,
and those two disjoint orbits have sizes at least $k$ and $\ell$.
This contradicts the assumed size. The interval-diagram statement
uses the faithful signed channel action of
Lemma~\ref{lem:unit-channels}.
\end{proof}

\section{Permutation queries and the lower bounds}\label{sec:lower}

\subsection{Fixed answer continuations}

For an unknown $\pi\in\Alt_d$, consider a \emph{query controller}.
Its initial instruction is a query index or a terminal outcome. Its
answer continuation is a fixed function
\[
 t:[d]\longrightarrow[d]\sqcup\{\acc,\rej\}.
\]
An instruction $a\in[d]$ obtains the raw answer $\pi(a)$ and then
executes $t(\pi(a))$. Thus the answer retains neither the caller nor
a separate copy of its query index. The terminal instructions accept
and reject; an infinite sequence of queries rejects. Fixed helper
computations can be folded into the instructions, with a fixed helper
loop counted as rejection. The first answer uses the same function
$t$ as every later answer.

\begin{lemma}[From one channel orbit to queries]\label{lem:query-interface}
Suppose a subgroup of the interval-diagram group $eMe$, with semantic
parameter $\pi\in\Alt_d$, acts naturally on exactly one orbit of
$d$ signed through channels and fixes all other channels. The natural
action is that of the same semantic permutation $\pi$. After fixing
ordinary left and right contexts, acceptance on the corresponding
middle-word representatives is decided by a query controller. Its raw
operation is $a\mapsto\pi(a)$ when the moving orbit is forward, and
$a\mapsto\pi^{-1}(a)$ when it is backward.
\end{lemma}
\begin{proof}
Choose one nonempty word for each group element. The return functions,
failure domains, and through domains are fixed by
Lemma~\ref{lem:unit-channels}. A forward through row has the form
$F_h(q)=\alpha_h(F_e(q))$, and a backward row has the form
$B_h(q)=\beta_h(B_e(q))$. Outside the moving orbit these responses
are constant. On that orbit, the forward response is $\pi(a)$ or
the backward response is $\pi^{-1}(a)$, because the homomorphic
channel action is $(\alpha_h,\beta_h^{-1})$.

Consider the left region consisting of the actual left endmarker and
the fixed left context, and the right region consisting of the fixed
right context and the actual right endmarker. Replace each visit to
the middle interval by its complete exit response. Fold the two fixed
regions, all return rows, and all constant through responses into a
finite deterministic network. From a boundary configuration, follow
this network until it accepts, rejects, enters a fixed cycle, or asks
for a variable through response. In the last case its query index is
fixed by that boundary configuration. Failure or a cycle gives the
reject instruction.

A variable forward answer is a physical state at the first cell of
the fixed right region. A variable backward answer is a physical
state at the last cell of the fixed left region. Each natural-orbit
label determines one such state, and all variable answers use the
same side because the orbit lies wholly in $U$ or wholly in $V$.
The continuation therefore depends only on the answer label. This
also covers empty ordinary contexts, since their regions still contain
the corresponding endmarkers. Any fixed processing of a raw answer
is included in its one continuation.

Start the same folding from the actual initial configuration. It
either has a fixed terminal outcome or reaches a first query. The
first raw answer arrives at exactly the same answer configuration as
later occurrences of that label, with no initialization tag. Although
middle representatives can have different lengths, positions inside
each fixed region are identified relative to its boundary. The finite
control has no clock or absolute-position register that could retain
traversal length. These relative configurations and their transition
rules are identical across representatives. The resulting initial
instruction and fixed answer function give the stated controller,
including all rejecting loops.
\end{proof}

\subsection{An inverse-query obstruction}

\begin{lemma}[Even completion]\label{lem:even-completion}
For $d\ge5$, every prescription of at most three distinct
input--output pairs with distinct inputs and distinct outputs
extends to an even permutation of $[d]$.
\end{lemma}
\begin{proof}
Complete the prescribed injection arbitrarily to a permutation.
If it is odd, interchange the images of two unassigned inputs.
There are at least $d-3\ge2$ such inputs; this changes the parity
without changing a prescribed pair.
\end{proof}

\begin{lemma}[A limit on permutation queries]\label{lem:inverse-obstruction}
For $d\ge5$, a query controller cannot recognize
\begin{equation}\label{eq:inverse-test}
                            \pi^{-1}(i)\in T
\end{equation}
when $i\in[d]$ and $3\le|T|\le d-2$. Infinite runs count as
rejection, as in the automaton model.
\end{lemma}
\begin{proof}
The predicate is nonconstant, so the initial fixed computation must
query some index $a$. For each first answer $j\ne i$, the condition
$\pi(a)=j$ is compatible with either truth value in \eqref{eq:inverse-test}: prescribe
the preimage of $i$ in $T\setminus\{a\}$ for an instance satisfying the predicate,
or in $([d]\setminus T)\setminus\{a\}$ for one that does not.
Both sets are nonempty. These are two-pair prescriptions and hence
have even completions by Lemma~\ref{lem:even-completion}.

Thus no answer $j\ne i$ may accept, reject, or follow only a fixed
helper loop. It must query an index $f(j)$. If $a\notin T$, the
first answer $i$ occurs on a negative input, so its instruction may
not accept either. There would then be no accepting instruction at
any answer control. A finite accepting run would be impossible,
although positive inputs exist. Therefore $a\in T$. Since the
only answer that could accept is $i$, its instruction must accept.
Otherwise none of the answer controls could terminate with acceptance.

If $f(j)=a$ for some $j\ne i$, choose any $c\in T\setminus\{a\}$
and an even permutation satisfying $\pi(a)=j,\pi(c)=i$.
The predicate holds for this permutation, but the controller repeats
query $a$ and answer $j$ forever. Thus $f(j)\ne a$. If $f(j)\notin T$, prescribe
$\pi(a)=j$ and $\pi(f(j))=i$. The computation accepts after the
second query, although the prescribed preimage of $i$ is outside
$T$. Hence necessarily
\begin{equation}\label{eq:next-query}
                         f(j)\in T\setminus\{a\}
                              \quad(j\ne i).
\end{equation}

There are $d-1$ arguments $j\ne i$ but fewer than $d-1$ possible
values in \eqref{eq:next-query}, so two distinct $j_1,j_2\ne i$ have
$f(j_1)=f(j_2)=b$. Choose $c\in T\setminus\{a,b\}$, possible
because $|T|\ge3$. Lemma~\ref{lem:even-completion} supplies an even
permutation with
\[
                    \pi(a)=j_1,\qquad\pi(b)=j_2,\qquad\pi(c)=i.
\]
It satisfies \eqref{eq:inverse-test}, but its run asks $a$, then $b$, then $b$
forever, and hence rejects. This is the required contradiction.
\end{proof}

\subsection{The complement lower bound}

\begin{theorem}\label{thm:complement-lower}
For $k\ge11$ and $\ell\in\{k,k+1\}$, every 2DFA recognizing
$\overline{L_{k,\ell}}$ has at least $2(k+\ell)$ states.
\end{theorem}
\begin{proof}
Suppose a target for the complement has $m<2(k+\ell)$ states. Lemmas~\ref{lem:semantic-product}, \ref{lem:group-extraction},
\ref{lem:unit-channels}, and~\ref{lem:capacity} give a faithful group of signed channels mapping onto $\Alt_k\times\Alt_\ell$, with at most $m$ channels. Lemma~\ref{lem:isolation} supplies a subgroup mapping isomorphically onto one semantic factor, with exactly one natural moving orbit for that factor, and fixing all other channels. Its other semantic coordinate is the identity. Every element has an actual nonempty word representative.

We now choose contexts that turn acceptance into the query test of
Lemma~\ref{lem:inverse-obstruction}. Each transformation used in
these contexts is singular, so the contexts exist by
Lemma~\ref{lem:paired-generation}. In an inactive coordinate use the rank-two idempotent $\eta_d$ that fixes $1,2$ and sends every other point to $2$. It need not act as identity on all points. For each active coordinate of degree $d$, a selector with
$a(1)=x,a(2)=y$ can be chosen to send every index other than $1$
to $y$. An output test for $S\subseteq[d]$ sends $S$ to $1$ and
its complement to $2$. Thus each active map also has rank at most
two. All these maps are singular because $d\ge11$.

\emph{The subgroup maps onto the forward factor.} The semantic middle pair is $(\pi,\id)$, with $\pi\in\Alt_k$. Choose left and right forward transformations $a,b$ with $a(1)=x$, $a(2)=y$, and $b^{-1}(1)=S$. Both backward context transformations are $\eta_\ell$. The entire backward map fixes $1,2$, so the source answer is
\begin{equation}\label{eq:forward-factor}
\pi(x)\in S\quad\hbox{or}\quad\pi(y)\in S.
\end{equation}
If the moving channel orbit is forward, take $x\ne y$ and $S=\{i\}$. The complement predicate is $\pi^{-1}(i)\in[k]\setminus\{x,y\}$, forbidden by Lemmas~\ref{lem:query-interface} and \ref{lem:inverse-obstruction}.
If the orbit is backward, set $x=y=i$ and $|S|=2$. Put $\sigma=\pi^{-1}$. The raw query is $\sigma$, and the complement predicate is $\sigma^{-1}(i)\in[k]\setminus S$, again forbidden.

\emph{The subgroup maps onto the backward factor.} Write the semantic middle pair as $(\id,\sigma)$, so the actual backward map is $g=\sigma^{-1}$. Give both forward contexts transformation $\eta_k$. The entire forward map fixes $1,2$. Put a selector $a$, with $a(1)=x,a(2)=y$, in the \emph{right} backward context, and an output test $b$, with $b^{-1}(1)=S$, in the \emph{left} backward context. The source answer is
\begin{equation}\label{eq:backward-factor}
g(x)\in S\quad\hbox{or}\quad g(y)\in S.
\end{equation}
If the moving orbit is forward, its raw query is $\sigma$. Set $x=y=i$ and $|S|=2$; the complement is $\sigma^{-1}(i)\in[\ell]\setminus S$.
If the orbit is backward, its raw query is $g=\sigma^{-1}$. Set $x\ne y$ and $S=\{i\}$; the complement is $g^{-1}(i)\in[\ell]\setminus\{x,y\}$.
Both predicates violate the same obstruction. This covers every possible isolated factor and orientation. Therefore $m\ge2(k+\ell)$.
\end{proof}

\subsection{Exact source complexity}

\begin{corollary}\label{cor:source-minimal}
For $k\ge11$ and $\ell\in\{k,k+1\}$,
$\statecomplexity(L_{k,\ell})=k+\ell$.
\end{corollary}
\begin{proof}
Let an arbitrary $m$-state automaton recognize $L_{k,\ell}$. The
semantic separation in Lemma~\ref{lem:semantic-product} applies to
this recognizer as well as to a complementing target. Its interval
diagram semigroup therefore maps onto $\Alt_k\times\Alt_\ell$.
Lemmas~\ref{lem:group-extraction} and~\ref{lem:unit-channels} give
a group with a faithful action on the signed channels $U\sqcup V$.
Corollary~\ref{cor:product-capacity} gives $|U|+|V|\ge k+\ell$,
while Lemma~\ref{lem:capacity} gives $|U|+|V|\le m$. Hence
$m\ge k+\ell$. The five-tuple in Section~\ref{sec:new-witness}
attains this bound.
\end{proof}

\section{Halting complementation and incoming directions}
\label{sec:direction-upper}

\subsection{Direction-determinate sources}

We give the reverse-tree traversal explicitly in the present model.
The method goes back to Sipser~\cite{Sipser1980}; its use in
complementation appears in GMP~\cite{GMP}. The two-state-per-source-state
traversal for direction-determinate graph-walking automata is given
by Kunc and Okhotin~\cite{kuncokhotin2020} and restated in
\cite[Lemma~8]{martynovaokhotin2021}. We specify the initial state and accepting configuration as part of
the transition table, so neither requires an additional state.

Call a source automaton \emph{direction-determinate} if there is a map
$d:Q\to\{-1,+1\}$ such that every defined transition, including every
endmarker transition, has the form
\[
  \delta(p,a)=(q,d(q)).
\]
Thus the state entered by a transition determines its direction. The two banks
of our witness have this property: every $P_i$ is entered moving right, and
every $Q_j$ is entered moving left.

\begin{theorem}\label{thm:direction-upper}
Every direction-determinate $n$-state source has a halting complement
with at most $2n$ states in the exact model of this paper.
\end{theorem}

\begin{proof}
Write the source as $(Q,\Sigma,\delta,s,f)$. If $s=f$, its language is
$\Sigma^*$; a two-state automaton with a nonfinal initial state and no
transitions recognizes its complement. If $s\ne f$ and $d(f)=+1$, the
configuration $(f,\LL)$ is unreachable, and a one-state automaton whose
initial state is final recognizes the complement. Both constructions halt and
use at most $2n$ states. Hence assume $s\ne f$ and $d(f)=-1$.

Let $\delta^*$ be $\delta$ with the row $(f,\LL)$ deleted. Acceptance takes
priority over this row, so its deletion does not change the source language.
Fix a total order $<$ on $Q$. The complement has state set
\[
  Q'=\{B_q:q\in Q\}\mathbin{\dot\cup}\{F_q:q\in Q\},
  \qquad A^c=(Q',\Sigma,\delta',B_f,F_f).
\]
These are all the states of the complement. Their roles are as follows. A configuration $(B_q,j)$
enumerates the predecessors of the source configuration $(q,j+d(q))$;
all those predecessors lie at cell $j$. A configuration $(F_p,j)$ means
that the reverse subtree rooted at $(p,j)$ has been traversed, and the
construction should advance to the next sibling or return to the parent.

Here is a complete finite prescription for the transition table. For each
ordinary symbol or endmarker $a$, put
\[
  P(q,a)=\{p\in Q:\delta^*(p,a)=(q,d(q))\},
\]
listed in increasing order. Given any increasing sublist $C$ of $Q$, define
$\operatorname{select}(C,a)$ by inspecting its candidates in order:
\begin{enumerate}
\item If $a=\LL$ and the current candidate $p$ is $s$, return
$\mathsf{reject}$.
\item If $a=\LL$ and $d(p)=+1$, or if $a=\RR$ and $d(p)=-1$,
skip candidate $p$. Such a source configuration has no predecessor, since
its predecessor cell would be outside the tape.
\item Otherwise return the transition $(B_p,-d(p))$.
\end{enumerate}
If the list is exhausted, return $\mathsf{exhausted}$. The selection is made when the transition table is constructed,
not during the run. It therefore uses no extra states or stationary
moves. In particular, the initial-configuration
test precedes the boundary-leaf test.

Set the initialization row to
\[
  \delta'(B_f,\LL)=(B_f,+1).
\]
Apart from this exceptional row, leave $(B_q,\LL)$ undefined when
$d(q)=-1$, and leave $(B_q,\RR)$ undefined when $d(q)=+1$.
For every other backward row $(B_q,a)$, evaluate
$\operatorname{select}(P(q,a),a)$. Use its returned transition;
$\mathsf{reject}$ means that the row is undefined, and
$\mathsf{exhausted}$ gives the transition $(F_q,d(q))$.

For a forward row $(F_p,a)$, the accepting configuration $(F_f,\LL)$
needs no transition. If $\delta^*(p,a)$ is undefined, leave the row undefined.
Otherwise write $\delta^*(p,a)=(q,d(q))$ and evaluate
\[
  \operatorname{select}\bigl(\{p'\in P(q,a):p'>p\},a\bigr).
\]
Again use its returned transition, leave the row undefined on
$\mathsf{reject}$, and use $(F_q,d(q))$ on $\mathsf{exhausted}$.
Every unspecified row is undefined.

All defined endmarker moves are inward. A selected backward descent that
would move outward is eliminated by the boundary-leaf test. An exhausted
forward row follows the source's original inward marker move. An exhausted
backward row moves to its parent, and the two cases in which that parent
would lie outside the tape were explicitly left undefined. Furthermore,
every transition entering $B_f$, except its initialization row, moves
$-d(f)=+1$. Consequently $(B_f,\LL)$ cannot be revisited, which justifies
reusing $B_f$ as the initial state without an additional state.

To prove correctness, fix an input. Consider the source configurations whose
runs under $\delta^*$ reach $(f,\LL)$. They form a finite tree directed
toward the root $(f,\LL)$: each nonroot vertex has one successor, and no
directed cycle can reach the root. Order the children of each vertex by
source-state order. Incoming direction determinacy places all children of
one vertex at the same input cell.

The prescribed machine performs depth-first traversal of this reverse tree.
A backward row starts a child list; a forward row advances in that list or
returns to the parent. Boundary leaves are skipped without entering them,
after checking whether they are the source's initial configuration. The
machine rejects immediately if it encounters $(s,\LL)$. If it finishes
the tree without this encounter, it reaches $(F_f,\LL)$ and accepts.
Thus it rejects exactly when the source accepts. The finite tree traversal
always terminates, even when the source's initial computation loops. This
proves both complementation and halting with $2n$ states.
\end{proof}

\begin{proof}[Proof of Theorem~\ref{thm:main}]
Take $k=\lfloor n/2\rfloor$, $\ell=\lceil n/2\rceil$, and
$K_n=L_{k,\ell}$. Corollary~\ref{cor:source-minimal} gives
$\statecomplexity(K_n)=k+\ell=n$. The witness is direction-determinate,
since all $P$-states are entered moving right and all $Q$-states
moving left. Theorem~\ref{thm:direction-upper} therefore supplies
a halting complement with at most $2n$ states.
Theorem~\ref{thm:complement-lower} requires at least $2n$ states.
Thus $\statecomplexity(\overline{K_n})=2n$.
\end{proof}

\subsection{States entered in both directions}

\begin{corollary}\label{cor:mixed-directions}
For an arbitrary $n$-state source, first delete the effective outgoing row at
$(f,\LL)$. Let $b$ be the number of states that are targets of both a
left-moving transition and a right-moving transition in the resulting table,
including endmarker transitions. Its complement has a halting recognizer
with at most $2n+2b$ states.
\end{corollary}

\begin{proof}
The cases $s=f$, or $s\ne f$ with no transition entering $f$ while moving
left, have the trivial complements used above. In every other case, split
each of the $b$ mixed states $q$ into $q^-$ and $q^+$. Keep one copy of every
other state, assigned its unique incoming direction, or any direction if it
has no incoming transitions. Give every copy of a source state the same
outgoing rows; a row whose original target was $q$ and whose direction was
$\varepsilon$ enters the copy of $q$ assigned direction $\varepsilon$.
Choose any existing copy of $s$ as initial and the copy of $f$
assigned direction $-1$ as final; denote this final copy by $f^-$. The resulting machine has $n+b$ states and is
direction-determinate.

Projection of the cloned states to their original names reproduces each
run of the original source. Since $s\ne f$, every visit to $(f,\LL)$
must have arrived by a left move and therefore corresponds precisely to
$(f^-,\LL)$. Thus the splitting preserves the language, including infinite
rejecting runs. Theorem~\ref{thm:direction-upper} gives the stated bound.
\end{proof}

In particular, a family requiring at least $cn$ complement states for a
fixed $c>2$ must have
\[
  b\ge \frac{c-2}{2}\,n.
\]
Thus a family with only $o(n)$ mixed states has complement complexity
at most $2n+o(n)$. Since always $b\le n$, the same construction
also proves $C_3(n)\le4n$ in the exact model of this paper. Together
with Theorem~\ref{thm:main}, this gives
\[
 2n\le C_3(n)\le4n\qquad(n\ge22).
\]
The matching $2n$ upper bound applies to direction-determinate sources;
the lower bound in Theorem~\ref{thm:main} applies to unrestricted targets.

\section{Comparison with GMP's acceptance convention}\label{sec:conventions}

The exact equalities of Theorem~\ref{thm:main} use left-endmarker
acceptance. In GMP's definition~\cite[Section~2]{GMP}, stationary
moves are allowed, there may be several final states, and a computation
accepts upon reaching any final state at any tape position. Write
$\operatorname{sc}_{\mathrm{GMP}}(L)$ for the minimum number of states
under that convention. Our state $Q_1$ participates in ordinary
backward scans, so it cannot simply be declared globally final.

This difference concerns the role of a final state, rather than the
choice of the left or right endmarker. In our model the final state
can do ordinary work away from the accepting marker. GMP's
Lemma~3.1~\cite[pp.~1176--1177]{GMP} normalizes a nontrivial-language
recognizer without adding states: it has one final state, that state
sweeps left on every symbol other than $\LL$, its left-marker row
is undefined, and there are no stationary moves. A visit to that
state therefore reaches the left endmarker. Reading the normalized
table with our acceptance condition preserves the language.

Conversely, an $r$-state automaton $A=(Q,\Sigma,\delta,s,f)$ in
our model has an $(r+1)$-state GMP realization
\[
 \widehat A=(Q\mathbin{\dot\cup}\{a\},\Sigma,\widehat\delta,s,\{a\}),
\]
where $a$ is fresh, $\widehat\delta(f,\LL)=(a,0)$, all rows of
$a$ are undefined, and every other old row is unchanged. Only $a$
is globally final. Thus for every nontrivial language $L$,
\[
 \statecomplexity(L)\le \operatorname{sc}_{\mathrm{GMP}}(L)
              \le \statecomplexity(L)+1.
\]
The extra state does not affect the leading factor, but it matters
for the exact bound. Merely applying these two
inequalities gives a $2N-2$ complement lower bound for the
$N=n+1$ state source. A fixed channel in GMP's normal form yields
the stronger exact comparison below.

\begin{lemma}[A fixed accepting channel]\label{lem:gmp-fixed-channel}
Let an $m$-state automaton be in the normal form of GMP's
Lemma~3.1: its unique final state $f$ moves left on every symbol
other than $\LL$, and its row at $\LL$ is undefined. Interpret its
acceptance at $\LL$ only. For any group of nonempty interval
diagrams obtained as in Lemma~\ref{lem:group-extraction}, the signed
channel action has a globally fixed backward channel. Removing this
channel leaves a faithful action on at most $m-1$ points.
\end{lemma}
\begin{proof}
For every nonempty ordinary word $w$, a right entry in state $f$
sweeps left through the entire interval and exits in state $f$.
Hence $B_w(f)=f$. In particular, for the identity diagram $e$,
we have $f\in V=\im B_e$. For every group element $h$, the normal
form $B_h=\beta_h B_e$ gives
\[
 \beta_h(f)=\beta_h(B_e(f))=B_h(f)=f.
\]
Thus the signed point $(\mathsf{backward},f)$ is fixed by the
homomorphic action $(\alpha_h,\beta_h^{-1})$.
Its removal preserves faithfulness: any element fixing all remaining
points also fixes the removed point, and hence is the identity.
By Lemma~\ref{lem:capacity}, the remaining set has size
$|U|+|V|-1\le m-1$.
\end{proof}

\begin{theorem}[Exact gap under GMP's convention]\label{thm:gmp-exact}
For every $N\ge23$, there is a language $L$ over the fixed ternary
alphabet $\Sigma$ such that, in GMP's model,
\[
 \operatorname{sc}_{\mathrm{GMP}}(L)=N,
 \qquad
 \operatorname{sc}_{\mathrm{GMP}}(\overline L)=2N-1.
\]
In particular, the gap asked for in
\cite[Section~6, p.~1186, problem~(e)]{GMP} exists.
\end{theorem}
\begin{proof}
Put $n=N-1\ge22$ and $L=K_n$. These languages are nontrivial:
$\varepsilon\in K_n$ and $\mathtt C\notin K_n$.

For the source upper bound, add a fresh state $a$ to the $n$-state
source five-tuple, make $a$ the sole globally final state, and set
$\delta(Q_1,\LL)=(a,0)$. Retain all other original rows and leave
all rows of $a$ undefined. No original state is globally final.
This is an $(n+1)$-state GMP automaton for $K_n$.
For the complement upper bound, take the $2n$-state halting complement
from Theorem~\ref{thm:direction-upper}. Apply the same conversion:
add a fresh sole final state $a$, replace the accepting configuration's
unused left-marker row by a stationary move into $a$, and keep every
other original state nonfinal. This gives a $(2n+1)$-state GMP
automaton for $\overline{K_n}$.

For either lower bound, normalize an arbitrary $m$-state GMP
recognizer using GMP's Lemma~3.1, with no increase in states.
Because its final state sweeps left, interpreting finality only at
$\LL$ preserves the language. Contextual separation still gives
a group of actual interval diagrams mapping onto
$\Alt_k\times\Alt_\ell$, where $k+\ell=n$.
By Lemma~\ref{lem:gmp-fixed-channel}, this group has a faithful
action on at most $m-1$ signed channels after deleting its fixed
accepting channel.

For a recognizer of $K_n$, Corollary~\ref{cor:product-capacity}
requires at least $k+\ell=n$ points in this remaining action.
Consequently $m\ge n+1$.

For a recognizer of $\overline{K_n}$, suppose that $m<2n+1$.
The remaining faithful action has fewer than $2(k+\ell)$ points,
so Lemma~\ref{lem:isolation} supplies one natural moving orbit for
one semantic factor, with all other remaining channels fixed.
Restoring the deleted channel adds only a fixed point. The subgroup
therefore satisfies exactly the hypotheses of the query reduction
in Theorem~\ref{thm:complement-lower}. Its four factor-and-direction
cases give the same contradiction. Hence $m\ge2n+1$.

The matching upper and lower bounds establish the equalities
$n+1=N$ and $2n+1=2N-1$.
\end{proof}

For this family the GMP ratio is $(2N-1)/N=2-1/N$, tending to two.
The explicit $(2N-1)$-state complement also shows that the same
language cannot witness a $2N$ lower bound in GMP's convention.
Theorems~\ref{thm:main} and~\ref{thm:gmp-exact} give the exact
counts in the two respective models.

\section{Conclusion}

We have exhibited a language of two-way deterministic state complexity
$n$ whose complement requires exactly $2n$ states, for every $n\ge22$,
over one fixed ternary alphabet. In GMP's convention the same family
has exact complexities $N$ and $2N-1$. This answers their question
about the existence of a complementation gap.

The witness reaches the $2n$ upper bound for direction-determinate
sources. For general sources, the bound $2n+2b$ shows that a larger
leading constant requires linearly many states that can be entered
from both directions. The optimal general constant remains between
two and four.

\section*{Acknowledgments}
I thank my advisor, Christos Kapoutsis, for his guidance and support.
I would also like to thank Adam Ghaida for his help verifying the
correctness of my proof.

This work was developed with assistance from OpenAI's Codex Astra Ultra.
I contributed initial ideas and discussed the constructions and proofs
with Astra. Astra developed most of my mathematical outputs, including the two-bank witness and its ternary encoding, the
interval-diagram and alternating-group arguments, and the
permutation-query obstruction. Astra also helped adapt the known
reverse-tree traversal to the stated model, derive the comparison
with GMP's convention, and write and run finite checks. I take responsibility for the manuscript and
any errors.

\bibliographystyle{plain}
\bibliography{complementation_exact_references}
\end{document}